\documentclass[12pt,doublespaced]{amsart}
\usepackage{latexsym,amsfonts,amsmath,amssymb,amsthm}
\usepackage[all]{xy}

\newcommand{\IR}{{\mathbb{R}}}

\newcommand{\ZZ}{{\mathbb{Z}}}

\newcommand{\Om}{{\Omega}}

\newcommand{\vfi}{{\varphi}}

\newcounter{smalllist}

\newtheorem{theorem}{Theorem}

\newtheorem{lemma}{Lemma}[section]
\newtheorem{prop}[lemma]{Proposition}

\theoremstyle{definition}

\newtheorem{remark}[lemma]{Remark}

\let\Re=\undefined\DeclareMathOperator{\Re}{Re}

\let\llldots=\ldots
\def\ldots{\llldots{}}

\numberwithin{equation}{section}

\begin{document}

\title[Self-adjointness on bounded domains]{On confining potentials and essential self-adjointness for Schr\"{o}dinger operators on
bounded domains in $\IR^n$}
\author[Gh.~Nenciu and I.~Nenciu]{Gheorghe Nenciu and Irina Nenciu}
\address{Gheorghe Nenciu\\
Department of Theoretical Physics and Mathematics, University of
Bucharest\\
P.O. Box MG 11, RO-077125 Bucharest, Romania \textit{and}
         Institute of Mathematics ``Simion Stoilow''
     of the Romanian Academy\\ 21, Calea Grivi\c tei\\010702-Bucharest, Sector 1\\Romania}
\email{Gheorghe.Nenciu@imar.ro}
\address{Irina Nenciu\\
         Department of Mathematics, Statistics and Computer Science\\
         University of Illinois at Chicago\\
         851 S. Morgan Street\\
         Chicago, IL \textit{and} Institute of Mathematics ``Simion Stoilow''
     of the Romanian Academy\\ 21, Calea Grivi\c tei\\010702-Bucharest, Sector 1\\Romania}
\email{nenciu@math.uic.edu}

\begin{abstract}
Let $\Om$ be a bounded domain in $\IR^n$ with $C^2$-smooth boundary, $\partial\Om$,
of co-dimension 1, and let $H=-\Delta +V(x)$ be a Schr\"odinger operator on $\Om$ with potential $V \in L^{\infty}_{loc}(\Om )$. We seek the weakest conditions we can find on the rate of growth of the potential $V$ close to the boundary $\partial\Om$ which guarantee essential self-adjointness of $H$ on $C_0^\infty(\Om)$ .  As a special case of an abstract condition, we add optimal logarithmic type corrections to the known condition  $V(x)\geq \frac{3}{4d(x)^2}$ where  $d(x)=\text{dist}(x,\partial\Om)$. More precisely,
 we show that
if, as $x$ approaches $\partial\Om$,
$$
V(x)\geq \frac{1}{d(x)^2}\biggl(\frac34-\frac{1}{\ln(d(x)^{-1})}-
\frac{1}{\ln(d(x)^{-1})\cdot\ln\ln(d(x)^{-1})}-\cdots\biggr)
$$
where  the brackets
contain an arbitrary finite number of logarithmic terms, then $H$
is essentially self-adjoint on $C_0^\infty(\Om)$. The constant 1
in front of each logarithmic term is optimal. The proof is based
on a refined Agmon exponential estimate combined with a well known multidimensional Hardy
inequality.
\end{abstract}
\thanks{We wish to thank F.~Gesztesy, A.~Laptev, M.~Loss and B.~Simon for useful comments and suggestions. I.N.'s research was partly supported by the NSF grant DMS 0701026.}
\maketitle

\tableofcontents

\section{Introduction}\label{S:1}

Consider a particle in a bounded domain $\Om$ in $\IR^n$, $n\geq1$, in the presence of a potential $V$.
At the heuristic level,
if $V(x)\to\infty$ as $x$ approaches the boundary $\partial\Om$, then the particle is confined
in $\Om$ and never visits the boundary.
At the classical level, this indeed happens when $V(x)\to\infty$ as $x\to\partial\Om$
(see, e.g. \cite[Theorem~X.5]{ReeSim}). At the quantum level, the problem is
much more complicated due to the possibility that the particle tunnels through the infinite
potential barrier and ``sees'' the boundary. The fact
that the particle never feels the boundary amounts to saying that $V$ determines completely
the dynamics: there is no need for boundary
conditions. At the mathematical level, by Stone's Theorem, the problem is then finding conditions
on the rate of growth of $V(x)$ as $x\to\partial\Om$
which ensure that the Schr\"{o}dinger operator
\begin{equation}\label{E:I1}
H=-\Delta+V
\end{equation}
is essentially self-adjoint on $C_0^\infty(\Om)$. Let us note here that oscillations of the potential
could also play a role in the essential self-adjointness problem due to the possibility of coherent 
reflections by an appropriately chosen sequence of potential barriers (see \cite{ReeSim}, the 
Appendix to Chapter~X.1). In this paper we will
not consider oscillatory potentials, but rather focus on potentials which grow
to infinity at the boundary of the domain.

The problem has a long and distinguished history; for details and
further references, we send the reader to \cite{CL} and
\cite{ReeSim} and the review papers \cite{KSW, BMS, Bru}. In the
1-dimensional case (say, $\Om=(0,1)$) there exists a
well-developed theory of essential self-adjointness of
Sturm-Liouville operators, which is based on limit point/limit
circle Weyl type criteria (see e.g. \cite{CL}, \cite{ReeSim} and
the references therein). In particular if, under appropriate
regularity conditions,
\begin{equation}\label{E:I2}
V(x)\geq\frac34\cdot\frac{1}{d(x)^2}\,,
\end{equation}
where $d(x)=\text{dist}\bigl(x,\{0,1\}\bigr)$, then $H$ is essentially self-adjoint on $C_0^\infty(0,1)$.
The constant $\frac34$ is optimal, in the sense
that if for some $\varepsilon>0$,
$$
0\leq V(x)\leq \biggl(\frac34-\varepsilon\biggr)\cdot\frac{1}{d(x)^2}\,,
$$
near 0 and/or 1, then $H$ is not essentially self-adjoint on $C_0^\infty(0,1)$ (see Theorem~X.10 in
\cite{ReeSim}). Many results have been generalized
from one to higher dimensions -- see, for example, a comprehensive review of these results in \cite{Bru}.
In particular, if $\Om$ is a bounded domain with $C^2$ boundary $\partial\Om$
of codimension 1, and if $V$ satisfies \eqref{E:I2} as $x$ approaches
$\partial\Om$, with $d(x)=\text{dist}(x,\partial\Om)$,
then $H$ defined as in \eqref{E:I1} is essentially self-adjoint on $C_0^\infty(\Om)$.
Moreover, Theorem~6.2 in \cite{Bru}
implies that for the case at hand the essential self-adjointness of $H$ is assured by a
weaker condition, namely
\begin{equation}\label{E:I4}
V(x)\geq\frac34\cdot\frac{1}{d(x)^2} -\frac{c}{d(x)}\,
\end{equation}
with some $c\in\IR^+$. This raises the following optimality question:
While among power-type growth conditions, $\frac34\cdot\frac{1}{d(x)^2}$ is optimal both in the exponent
and in the constant, does a growth condition of the type
$$
V(x)\geq \frac34\cdot\frac{1}{d(x)^2}\biggl(1-m\bigl(d(x)\bigr)\bigg)\,,\qquad \lim_{t\to0+} m(t)=0,
\qquad m(t) \geq 0
$$
still imply essential self-adjointness of $H$? It turns out  that
this is false -- see the counterexample in the proof of
Theorem~\ref{T:pOptimal}. So the question of optimality should be
refined to asking whether $\frac34\cdot\frac{1}{d(x)^2}$ is the
leading term of a (possibly formal) asymptotic expansion near
$\partial\Om$ of a critical potential $V_c$ such that $V\geq V_c$
near $\partial\Om$ implies essential self-adjointness of $H$ on
$C_0^\infty(\Om)$. This would amount to finding the form and size
of sub-leading terms in the asymptotic expansion of $V_c$.

The main result of this note is the affirmative answer to this optimality question.
Namely, we show that for bounded domains $\Om$ in $\IR^n$, $n=1,2,3,...$
having $C^2$ boundary of codimension 1, and for potentials $V$ satisfying
\begin{equation}\label{E:I5}
V(x)\geq \frac{1}{d(x)^2}\biggl(\frac34-\frac{1}{\ln(d(x)^{-1})}-\frac{1}{\ln(d(x)^{-1})
\cdot\ln\ln(d(x)^{-1})}-\cdots\biggr)
\end{equation}
as $x$ approaches $\partial\Om$, the Schr\"odinger operator $H$ is
essentially self-adjoint on $C_0^\infty(\Om)$, and that the
constants 1 in front of each logarithmic term on the right-hand
side of \eqref{E:I5} are optimal (for a precise statement, see
Theorem~\ref{T:pOptimal}).

Two remarks are in order here. The first one is that we are
interested in optimality rather than generality. Accordingly, and
also in order not to obscure the main ideas of our proofs by
technicalities, we consider the simplest case, which is still the
most interesting from a physical point of view: a bounded domain
$\Om\subset\IR^n$ with $C^2$ boundary of co-dimension 1; in
addition, we only consider scalar Schr\"odinger operators with
regular ($L^\infty_\text{loc}$) potentials. In this setting
the proofs are short and elementary. At the price of
technicalities, one may be able to extend the results of the
present note to more general situations, e.g. boundaries with
components of higher co-dimension, local singularities of the
potential or second order elliptic operators of general
form. Reducing the regularity of the boundary $\partial\Om$ below
$C^2$ seems to require a finer analysis -- in particular, of
multidimensional Hardy inequalities on domains with 
less smooth boundaries (see e.g. \cite{Dav}, \cite{Dav1}, \cite{LS} and references 
therein for results in this area). In addition, we consider only
what one can think of as the ``isotropic'' case, i.e. we seek
conditions on $V(x)$ which depend only on $d(x)$, and not on the
specific point $x_0$ of the boundary that $x$ approaches, or the
direction along which $x\to x_0$.

The second remark concerns the method of proof. While the proofs
in \cite{Bru} are based on his theory of semimaximal operators,
our method of proof  is based on the observation that essential
self-adjointness follows (via the fundamental criterion for
self-adjointness, see, e.g., \cite{ReeSim}, \cite{AkhGla}) from
Agmon type results on exponential decay of eigenfunctions (see
\cite[Theorem~1.5a]{Agm}). As stated, the result in \cite{Agm}
does not lead to optimal growth conditions on the potential. One
has both to strengthen the exponential decay estimates, and to
combine them with multidimensional Hardy inequalities \cite{Dav}.
So our basic technical result is an exponential estimate of
Agmon-type -- see Theorem~\ref{T:2}. Here the point is that
our condition ($\Sigma.2$) below is strictly weaker than the
corresponding condition (3.12) from Brusentsev~\cite{Bru}.

The paper is organized as follows. In Section~\ref{S:2} we state
the problem and the main results. Section~\ref{S:3} contains the
proof of the Agmon-type Theorem~\ref{T:2}. While some of the
results in this section go back to Agmon \cite{Agm} and are
well-known (e.g. the identity in Lemma~\ref{L:2}), we give
complete proofs for the reader's convenience. Finally
Section~\ref{S:4} contains the proof of Theorems~\ref{T:1}
and~\ref{T:p}.

\section{Main results}\label{S:2}

Let $\Om$ be a bounded domain in $\IR^n$, $n\geq1$, with
$C^2$-smooth boundary, $\partial\Om$, of co-dimension 1. We
consider the function
\begin{equation}\label{D:dist}
d(x)=\text{dist}(x,\partial\Om),\qquad\text{for}\quad x\in\Om\,,
\end{equation}
where ``$\text{dist}$'' denotes the usual, Euclidean distance in
$\IR^n$. As is well-known (see, for example, the Appendix to
Chapter~14 in \cite{GilTru}), $d$ is Lipshitz and differentiable
a.e. in $\Om$. More importantly for us here, there exists a
constant
\begin{equation*}\label{E:212}
d_\Om>0\qquad \text{(depending only on the domain $\Om$)}
\end{equation*}
such that for $x\in\Om$ with $d(x)<d_\Om$, $d$ is twice-differentiable and
\begin{equation}\label{E:211}
|\nabla d(x)|\leq1\,.
\end{equation}

\begin{remark}
Actually  $|\nabla
d(x)|=1$  for $x\in\Om$ with $d(x)<d_\Om$, see for example
\cite{GilTru}, or Lemma 6.2 in \cite{Bru}, but in the proofs below
we use only \eqref{E:211}.
\end{remark}

In $\Om$ we consider the Schr\"odinger operator $H=-\Delta+V$ with
$V\in L^\infty_\text{loc}(\Om)$, defined on $\mathcal
D(H)=C_0^\infty(\Om)$. As explained in the Introduction, we are
seeking growth conditions on $V$ close to $\partial\Om$ ensuring
essential self-adjointness of $H$. These will be given in terms of
functions $G$ described below:

\textbf{Condition ($\Sigma$).} \textit{A function $G\,:\,
(0,\infty)\to \IR$ is said to satisfy condition
{\em(}$\Sigma${\em)} if it is $C^1(0,\infty)$ and such that:
\begin{enumerate}
\item[($\Sigma.1$)] There exists $d_0>0$, $d_0\leq d_\Om$, such that
\begin{equation*}\label{E:sigma1}
0\leq G'(t)\leq \frac1t,\qquad \text{for}\quad t\in(0,d_0)\quad\text{and}
\end{equation*}
\begin{equation*}\label{E:sigma2}
G'(t)=0\qquad \text{for}\quad t\geq d_0\,.
\end{equation*}
\item[($\Sigma.2$)] For any $\rho_0\leq \frac{d_0}{2}$,
\begin{equation}\label{E:sigma3}
\sum_{n=1}^\infty 4^{-n}e^{-2G(2^{-n}\rho_0)} =\infty
\end{equation}
\end{enumerate}}

We can now formulate our main result:
\begin{theorem}\label{T:1}
Consider an open, bounded domain $\Omega\subset\IR^n$ with $C^2$-smooth boundary, and the Schr\"{o}dinger
operator
\begin{equation}
H=-\Delta+V,
\end{equation}
with $V\in L^\infty_\emph{loc}$ and domain $\mathcal{D}(H)=C_0^\infty(\Omega)$.
Assume that there exists a function $G$ satisfying condition $(\Sigma)$ such that
\begin{equation}\label{E:2}
V=V_1+V_2,
\end{equation}
where
\begin{equation}\label{E:3}
V_1(x) +\frac14 \cdot\frac{1}{d(x)^2} \geq G'(d(x))^2\quad\text{and}\quad V_2\in L^\infty(\Om)\,.
\end{equation}

Then $H$ is essentially self-adjoint in $L^2(\Om)$.
\end{theorem}

Theorem~\ref{T:1} follows from the fundamental criterion for self-adjointness (see, for example,
 \cite{ReeSim},\cite{AkhGla}), a multidimensional Hardy inequality \cite{Dav}, and
a (refined) Agmon-type exponential estimate (see Theorem~\ref{T:2} in Section~\ref{S:3}).

We now turn to various examples of functions $G$ satisfying condition ($\Sigma$),
and the associated criteria for essential self-adjointness of $H$ in terms of the growth of the
potential at the boundary of the domain.

The first, simplest example of a function $G$ satisfying condition ($\Sigma$) is the one for which at sufficiently small $t$:
\begin{equation*}\label{Ex:11}
G(t)=\ln t\,,
\end{equation*}
which leads to the classical bound
\begin{equation}\label{Ex:12}
V_1(x)\geq \frac34\cdot\frac{1}{d(x)^2}\,,\qquad\text{as}\quad x\to\partial\Om\,.
\end{equation}

The second example is (again for  $t$ sufficiently small)
\begin{equation*}\label{Ex:21}
G(t)=\ln t-c\cdot t\,,\qquad c\in\IR^+\,.
\end{equation*}
This choice of $G$ leads, through \eqref{E:3}, to
\begin{equation}\label{Ex:22}
V_1(x)\geq \frac34\cdot\frac{1}{d(x)^2} - \frac{\tilde c}{d(x)}\,,\qquad\text{as}\quad x\to\partial\Om\,,
\end{equation}
for all $\tilde c <2c$. This is the lower bound obtained by Brusentsev in \cite[Theorem 6.2]{Bru}
for the case at hand.

The next example is (again for sufficiently small $t$) of the
form
\begin{equation*}\label{Ex:23}
G(t)=\ln t +\int_t f(u)\,du
\end{equation*}
with
\begin{equation}\label{Ex:24}
f(u)\geq 0,\quad \lim_{u\to0}uf(u)=0,\quad\text{and}\quad
\lim_{t\to0}\int_t f(u)\,du<\infty\,.
\end{equation}
This leads to a bound on $V$ of the form
\begin{equation}\label{Ex:25}
V(x)\geq \frac34\cdot\frac{1}{d(x)^2} -\frac{\tilde c}{d(x)}\cdot
f(d(x))\,, \qquad\text{as}\quad x\to\partial\Om\,,
\end{equation}
with $f$ as above and all $\tilde c <2$. Although this result does
not appear in an explicit form in \cite{Bru} it can still be
obtained from Corollary 3.3 in \cite{Bru}. Note that, since we
required that $uf(u)\to0$ as $u\to0$, the second term
$\frac{\tilde c}{d(x)}\cdot f(d(x))$ in \eqref{Ex:25} is of lower
order than $\frac{1}{d(x)^2}$, and thus does not contradict the
optimality of $\frac34\cdot\frac{1}{d(x)^2}$.

The last example is our main hierarchy of essential
self-adjointness conditions. Let $p\in\ZZ$, $p\geq 2$, and
iteratively define
\begin{equation}\label{E:DefnLp}
L_1(t)=\ln (1/t),\quad
L_p(t)=\ln L_{p-1}(t)\,,
\end{equation}
where each $L_p$ is defined for $t\in(0,e_p^{-1})$ with $e_1=e$
and $e_p=e^{e_{p-1}}$.  Then we have the following result:
\begin{theorem}\label{T:p}
Consider an open, bounded domain $\Omega\subset\IR^n$ with $C^2$-smooth boundary, and the Schr\"{o}dinger
operator
\begin{equation}
H=-\Delta+V,
\end{equation}
with $V\in L^\infty_\emph{loc}$ and domain $\mathcal{D}(H)=C_0^\infty(\Omega)$.
Let $p\in\ZZ$, $p\geq2$ and assume that
\begin{equation}\label{E:2p}
V=V_1+V_2,
\end{equation}
where
\begin{equation}\label{E:3p}
V_1(x)\geq \frac34\cdot\frac{1}{d(x)^2}-\frac{1}{d(x)^2}\sum_{j=2}^p \biggl(\prod_{k=1}^{j-1}
 L_k(d(x))\biggr)^{-1} -\frac{1}{d(x)}\cdot f(d(x))\,,
\end{equation}
for all $x$ with $d(x)<\min(e_p^{-1},d_\Om)$, with $f$ satisfying \eqref{Ex:24}, $V_1\geq 0$ on $\Om$,
and $V_2\in L^\infty(\Om)$.

Then $H$ is essentially self-adjoint in $L^2(\Om)$.
\end{theorem}

\begin{remark}  Let K be a positive constant. Rewriting $V(x)$ as
$$
V(x) = (V_1(x) +K) + (V_2(x)-K)
$$
one sees that it is sufficient to prove Theorem~\ref{T:p} with the
condition $V_1(x) \geq 0$ on  $\Om$ replaced by the condition
$V_1(x) \geq K$ on $\Om$ with K an arbitrary positive constant.
\end{remark}

Note that, for any given $j\geq 2$, each term $\frac1t \cdot
\biggl(\prod_{k=1}^{j-1} L_k(t)\biggr)^{-1}$ is non-integrable,
and hence a higher order correction than the integrable term
$f(t)$. Further note that the domain on which  $\sum_{j\leq p}
\biggl(\prod_{k=1}^{j-1} L_k(t)\biggr)^{-1}$ is well defined
shrinks to the empty set as $p\to\infty$.

The term $\frac14 \cdot\frac{1}{d(x)^2}$ in \eqref{E:3} comes from the 
additional ``barrier'' given by the uncertainty principle of quantum mechanics 
via the Hardy inequality (see \eqref{E:S4eq1} below).  The fact that Hardy 
inequalities appear here is not surprising since, as expressions of the 
uncertainty principle, they play a key role in various aspects of  the 
spectral analysis of 
Schr\"odinger and Dirac operators like stability, self-adjointness, etc 
(see e.g. \cite{G}, \cite{EL}, \cite{FLS}, \cite{HO2LT}  and the references therein). During the 
last decade a 
large body of literature about 
improvements to Hardy inequalities has appeared (see e.g. the references in \cite{DELV},\cite{FLS},\cite{BFT}, \cite{TZ}). 
In particular, in \cite{BFT} (under suitable conditions) the following optimal 
improvement of \eqref{E:S4eq1} was proved:
\begin{equation}\label{Hop}
\int_\Om\!\! |\nabla\varphi(x)|^2\,dx \geq
\frac14\int_\Om\!\!\frac{|\varphi(x)|^2}{d(x)^2}
\Bigl(1+\sum_{i=1}^{\infty} \prod_{k=1}^{i}X_{k}^{2}\biggl(\frac{d(x)}{D}\biggr) \Bigr)\,dx
\end{equation}
where $D$ is a sufficiently large constant, and $X_k(t)$, $t>0$
are defined recursively by
$$
X_1(t)=(1-\ln t)^{-1}, \quad X_k(t)=X_1(X_{k-1}(t)).
$$
However, this improvement of Hardy's inequality does not lead to an improvement of the result in 
Theorem ~\ref{T:p} (which according to Theorem \ref{T:pOptimal} is already 
optimal at the level of logarithmic subleading terms). Indeed, at the level of the
leading term, as $t\rightarrow 0$, $X_{k}^{2}=L_{k}^{-2}$ and so the 
contribution of the logarithmic terms in \eqref{Hop} can be absorbed in the 
last (integrable) term on the rigt-hand side of \eqref{E:3p}.

As we will show in Section~\ref{S:4}, the theorem follows from
Theorem~\ref{T:1} with  the following choice, for sufficiently
small $t$, of $G$ function:
 \begin{equation}\label{Ex:31}
 G_p(t)=\ln t +\frac12\cdot \sum_{j=2}^p L_j(t) +\int_t \tilde f(u)\,du\,,
 \end{equation}
 where $\tilde f$ also satisfies \eqref{Ex:24}.

Our last result is about the optimality of \eqref{E:3p}. 
With the hypotheses of Theorem~\ref{T:p}, it is well-know that the
constant $\frac34$ in front of the first term on the right-hand
side of \eqref{E:3p} is optimal. We claim that the constant 1 in
front of each logarithmic term in the sum above is also optimal, in the following precise
sense:

\begin{theorem}\label{T:pOptimal}
 Given $p\geq2$ and a constant $c>1$, there exist potentials
$V$ for which $H=-\Delta+V$ is \emph{not} essentially
self-adjoint, and which grow close to the boundary $\partial\Om$
as
\begin{equation}\label{E:4p}
\begin{aligned}
V(x)\geq& \frac34\cdot\frac{1}{d(x)^2}-\frac{1}{d(x)^2}\sum_{j=2}^{p-1}
\biggl(\prod_{k=1}^{j-1} L_k(d(x))\biggr)^{-1}\\
                      &-c\cdot\frac{1}{d(x)^2}\cdot\biggl(\prod_{k=1}^{p-1} L_k(d(x))\biggr)^{-1}\,.\\
\end{aligned}
\end{equation}
\end{theorem}

We end this section with a discussion of condition ($\Sigma$)
and its relation with condition (3.12) from Corollary 3.2 in
\cite{Bru}. We comment first on condition ($\Sigma .1$). Note that
 ($\Sigma.2$) implies that $G(t) \rightarrow -\infty$ as $t
 \rightarrow 0$. So $G'(t) \geq 0$ in ($\Sigma .1$) only adds that
$G(t) \rightarrow -\infty$ monotonically which is not a real
restriction as far as we are not considering (as already stated in
the Introduction) the effect of oscillations of the potential. In
fact, if one considers potentials which grow monotonically as $x
\rightarrow \partial \Omega $ one may impose even a stronger
condition that $G'(t)$ is monotonically increasing to $\infty$ as
$t \rightarrow 0$. Consider now $G'(t) \leq \frac{1}{t} $ in
($\Sigma .1$). This is again harmless (as far as it does not
contradict  ($\Sigma .2$)!) since if $G'_1(t)\geq G'_2(t)$ then
Theorem \ref{T:1} with $G(t)=G_2(t)$ gives a stronger  result than
with $G(t)=G_1(t)$.

The crucial condition is $(\Sigma .2)$ and this is to be compared
with Brusentsev's condition (3.12) from Corollary 3.3. We show now
that Brusentsev's condition (3.12) is (at least
for $G(t)$ satisfying $(\Sigma .1)$) strictly stronger than
$(\Sigma .2)$. Notice that we have restricted our attention to the
situation when his matrix $A\equiv I$. Comparing functions, we see
that in Brusentsev's notation the function which determines the
growth of the potential at the boundary is $\eta(x)$, and that we
are therefore interested in showing that, if
\begin{equation}\label{E:8sigma}
\eta(x)=-G(d(x)),
\end{equation}
satisfies condition (3.12) in \cite{Bru}, then $G$ must satisfy our
condition ($\Sigma.2$). Condition (3.12) in Brusentsev
guarantees that there exists a constant $C>0$ such that
\begin{equation}\label{brus}
|\nabla\eta(x)|\cdot e^{-\eta(x)}\leq C\,.
\end{equation}
If we recall that for $x$ with $d(x)$ small enough, $|\nabla
d(x)|=1$, then we get from \eqref{E:8sigma} and \eqref{brus} that
\begin{equation*}
\frac{d}{dt}\,e^{G(t)}=G'(t) e^{G(t)}\leq C,
\end{equation*}
for all $0<t< d_\Om$. But since $G(t)\to-\infty$ as $t\to0+$, we
can integrate, for all $n$ greater than some fixed integer
$N_\Om$,
\begin{equation*}
e^{G(2^{-n}\rho_0)}=\int_0^{2^{-n}\rho_0}  G'(t) e^{G(t)}\,dt\leq
2^{-n}\cdot C\rho_0\,.
\end{equation*}
Plugging this into the series from \eqref{E:sigma3} we get
\begin{equation*}
\sum_{n=1}^\infty 4^{-n}e^{-2G(2^{-n}\rho_0)}\geq \sum_{n \geq
N_{\Omega}}^\infty 4^{-n}\cdot 4^{n} (C\rho_0)^{-2}=+\infty\,,
\end{equation*}
thus showing that $G$ satisfies ($\Sigma.2$).

Conversely, recall the $G_p$ defined in \eqref{Ex:31}. As we will
show in Section~\ref{S:4}, the function $G_p$ satisfies
($\Sigma$). Take now the simplest case $G(t)=G_2(t)$ with $\tilde
f \equiv 0$ i.e. $G(t)=\ln t + \frac{1}{2} \ln \ln \frac{1}{t}$ for
sufficiently small $t$ and set
\begin{equation*}
\eta(x)=-G \bigl(d(x)\bigr)\,.
\end{equation*}
Then as $t=d(x) \rightarrow 0+$
\begin{equation*}
|\nabla\eta (x)|e^{-\eta(x)}=G'(t)\cdot e^{G(t)} =\biggl(\ln
\frac{1}{t}\biggr)^{\frac12}\biggl(1- \frac12\frac{1}{\ln \frac{1}{t}}\biggr)
\to+\infty\,,
\end{equation*}
and hence $\eta $ does not satisfy condition (3.12) from
\cite{Bru}.

\section{Agmon-type estimates}\label{S:3}

\begin{prop}\label{P:1}
Let  $\psi$ be a weak solution of
$$
H\psi=E\psi,
$$
i.e.$\psi\in H^1_\emph{loc}(\Om)$ and satisfies
\begin{equation}\label{E:4}
\langle\psi,(H-E)\varphi\rangle=0,\qquad\text{for every}\quad\varphi\in C_0^\infty(\Omega)
\end{equation}
Let $g\in C^1(\Omega)$ be a real-valued function for which there exists a constant $c>0$ such that
\begin{equation}\label{E:14}
\langle\varphi,(H-E)\varphi\rangle-\int_\Omega\!|\varphi(x)|^2|\nabla g(x)|^2\,dx\geq c\|\varphi\|^2
\end{equation}
for every $\varphi\in C_0^\infty(\Om)$.

For $\rho>0$, small enough, set
$\Om_\rho=\{x\in\Om\,|\,d(x)>\rho\}$. Then there exists a constant
$K=K(c)<\infty$, independent of $\rho$, such that
\begin{equation}\label{E:5}
\int_{\Om_{2\rho}}\!\! |e^{g(x)}\psi(x)|^2\,dx\leq
\frac{K(c)}{\rho}\int_{\Om_\rho\setminus\Om_{2\rho}}\!\!\biggl(\frac{1}{\rho}+|\nabla
g(x)|\biggr) |e^{g(x)}\psi(x)|^2\,dx\,.
\end{equation}
\end{prop}

Since this might be of independent interest and the proof is
the same, we will actually prove this proposition in a slightly
more general context. Indeed, consider the Schr\"{o}dinger
operator with magnetic potential on $\Om$
\begin{equation}\label{E:16}
H=(\vec{p}-\vec{a})^2+V,\quad V\in L^\infty_\text{loc}(\Om) ,\quad
\vec{a}\in C^1_\text{loc}(\Om) , \quad \vec{p}=-i \nabla,
\end{equation}
defined on $ \mathcal D(H)=C_0^\infty(\Omega)$ and, for $\varphi,\psi\in W^{1,2}$,
 the associated quadratic form
\begin{equation}\label{E:17}
h[\varphi,\psi]=\int_\Om\!\!\overline{(\vec{p}-\vec{a})\varphi}\cdot (\vec{p}-\vec{a})\psi\,dx
+\int_\Om \bar\varphi \cdot V\psi\,dx\,.
\end{equation}
Note that if $\varphi$ and $\psi$ are both in $C_0^2(\Om)$, then
$$
h[\vfi,\psi]=\int_\Om\!\! \overline{\vfi(x)}\,(H\psi)(x)\,dx\,.
$$

One of the main technical ingredients is the following simple identity~\cite{Agm}:
\begin{lemma}\label{L:2}
Let $\psi$ be a weak solution of $H\psi=E\psi$, and let $f=\bar f\in C_0^1(\Om)$. Then
\begin{equation}\label{E:18}
(h-E)[f\psi,f\psi]=\langle \psi, |\vec\nabla f|^2\psi \rangle\,.
\end{equation}
\end{lemma}

\begin{proof}
Consider first $f\in C_0^\infty$ and let $\vfi\in C_0^\infty$. Then
$$
(h-E)[\vfi,f\psi]=\langle(H-E)\vfi,f\psi\rangle=\langle f(H-E)\vfi,\psi\rangle\,.
$$
Since $[f,\vec p-\vec a]=i\nabla f$ on $C_0^\infty$, we get that
\begin{equation*}
[f,H]=[f,(\vec p-\vec a)^2]=i\bigl((\vec p-\vec a)\cdot\nabla f+\nabla f\cdot(\vec p-\vec a)\bigr),
\end{equation*}
and so, if we remember that $\psi$ is a weak solution,
\begin{equation*}
(h-E)[\vfi,f\psi]=\langle[f,H]\vfi,\psi\rangle=\langle\vfi,[H,f]\psi\rangle\,.
\end{equation*}
Since $f\psi\in W_0^{1,2}(\Om)$ and $C_0^\infty$ is dense in the $W^{1,2}$ topology,
the identity above implies that
\begin{equation}\label{E:19}
\begin{aligned}
(h-E)[f\psi,f\psi]
&=\langle\psi,f[H,f]\psi\rangle=\Re \langle\psi,f[H,f]\psi\rangle\\
&=\frac12\, \bigl\langle\psi,\bigl(f[H,f]-[H,f]f\bigr)\psi\bigr\rangle\\
&=\frac12\, \bigl\langle\psi,[f,[H,f]]\psi\bigr\rangle\,.
\end{aligned}
\end{equation}

Finally, a straightforward computation shows that
\begin{equation*}
[f,[H,f]]=-i[f,(\vec p-\vec a)\cdot\nabla f+\nabla f\cdot(\vec p-\vec a)]=
-i\,(2i\nabla f\cdot\nabla f)=2|\nabla f|^2\,,
\end{equation*}
which completes the proof.
\end{proof}

\begin{proof}[Proof of Proposition~\ref{P:1}]
As in~\cite{Agm}, we will now choose a function $f$ to plug
into the formula \eqref{E:18}. More precisely, let
$$f=e^g\phi,$$
where $g\in C^1(\Omega)$, real-valued, is the function from the statement of the proposition,
and $\phi\in C_0^\infty(\Om)$, $0\leq\phi\leq1$, is a cut-off function,
$$
\phi(x)=
\left\{
  \begin{array}{ll}
    0, & x\notin\Om_\rho \\
    1, & x\in\Om_{2\rho}.
  \end{array}
\right.
$$
Taking $\phi$ of the form $\phi (x)=k(d(x))$ where
$$
k(t)=
\left\{
  \begin{array}{ll}
    0, & 0 \leq t \leq \rho \\
    1, & t \geq 2\rho
  \end{array}
\right.
$$
one sees that for  $\rho$ small enough (say $\rho < \frac{d_{\Om}}{2}$)
\begin{equation}\label{E:10}
|\nabla \phi|\leq \frac{K_1}{\rho},
\end{equation}
with $K_1$ an absolute constant.
Then
$$
|\nabla f|^2=f^2|\nabla g|^2+m,
$$
where
$$
m=2fe^g\nabla g\cdot\nabla \phi+e^{2g}|\nabla\phi|^2.
$$
Estimating directly leads to:
\begin{equation*}
\begin{aligned}
|\langle\psi,m\psi\rangle|
&\leq \langle\psi,|m|\psi\rangle
=\int_{\Om}\!\!|\psi|^2\bigl(2e^{2g}\phi|\nabla g|\,|\nabla\phi|+e^{2g}|\nabla\phi|^2\bigr)\,dx\\
&\leq\frac{K_1}{\rho}\int_{\Om_\rho\setminus\Om_{2\rho}} \bigl|\psi e^g\bigr|^2\biggl(2|\nabla g|+
\frac{K_1}{\rho}\biggr)\,dx
\end{aligned}
\end{equation*}
where in the last inequality we used, as well as the fact that
$\nabla\phi\equiv0$ on
 $\bigl(\Om\setminus\Om_\rho\bigr)\bigcup\Om_{2\rho}$. But now
recall that the Agmon condition \eqref{E:14} was that
$$
(h-E)[\vfi,\vfi]-\int_\Omega\!|\vfi(x)|^2|\nabla g(x)|^2\,dx\geq c\|\varphi\|^2\,,
$$
with $c$ independent of $\vfi$ and $\rho$. Using the density of
$C_0^\infty$ in $W^{1,2}_0$, we obtain that
\begin{equation}\label{E:11}
(h-E)[f\psi,f\psi]-\langle f\psi,|\nabla g|^2f\psi\rangle\geq c\|f\psi\|^2\,.
\end{equation}
Since
$$
(h-E)[f\psi,f\psi]-\langle f\psi,|\nabla g|^2f\psi\rangle=\langle\psi,m\psi\rangle\,,
$$
we obtain
$$
\frac{K_1}{\rho}\int_{\Om_\rho\setminus\Om_{2\rho}} \bigl|\psi e^g\bigr|^2\biggl(2|\nabla g|+
\frac{K_1}{\rho}\biggr)\,dx
\geq|\langle\psi,m\psi\rangle|\geq c\int_\Om\bigl|f\psi\bigr|^2\,dx\,,
$$
which, if we recall the choice of $f$ made at the beginning of the
proof, leads directly to the claim of the proposition.
\end{proof}

\begin{theorem}\label{T:2}
Consider an open, bounded domain $\Omega\subset\IR^n$ with
$C^2$-smooth boun\-da\-ry, and the Schr\"{o}dinger operator
\begin{equation}\label{E:15}
H=-\Delta+V,
\end{equation}
with $V\in L^\infty_\emph{loc}$ and domain $\mathcal{D}(H)=C_0^\infty(\Omega)$.
Assume that there exist $E\in\IR$ and $c>0$ such that
\begin{equation}\label{E:13}
\langle \varphi,(H-E)\varphi\rangle-\int_\Om |\nabla g(x)|^2 |\varphi(x)|^2\geq c\|\varphi\|^2\,,
\end{equation}
for all $\varphi\in C_0^\infty(\Om)$, where $g(x)=G(d(x))$ for some $G$ satisfying condition $(\Sigma)$.

If $\psi$ is a weak solution of $H\psi=E\psi$, then $\psi\equiv0$.
\end{theorem}

\begin{proof}
Let $d_0>0$ be the constant that appears in condition $(\Sigma)$
for the function $G$ from the hypothesis. Fix, for the time being,
$0<\rho_0 \leq d_0/2$, and let $\rho>0$ be such that $2\rho \leq
\rho_0$. Then define a ``normalized'' $G$ function:
\begin{equation*}\label{E:20}
G_\rho(t)=G(t)-G(\rho)\,,
\end{equation*}
and set
$$
g_\rho(x)=G_\rho\bigl( d(x)\bigr).
$$
 Note that for all $x\in\Om$ we have
\begin{equation}\label{E:21}
\nabla g_\rho(x)=G'\bigl(d(x)\bigr)\nabla d(x).
\end{equation}
This, together with condition $(\Sigma.1)$ for $G$,  and the fact
that $|\nabla d(x)|\leq 1$ for $d(x)<d_\Om$, implies in particular
that
\begin{equation}\label{E:22}
\bigl|\nabla g_\rho(x)\bigr|\leq\frac{1}{d(x)}\qquad\text{for}\quad x\in\Omega\setminus\Omega_{d_0/2}\,.
\end{equation}
On the other hand, look at $x\in\Om_{\rho_0}$. Since
\begin{equation*}
\rho_0\leq d(x),
\end{equation*}
 condition ($\Sigma.1$) implies that
\begin{equation}\label{E:23}
g_\rho(x)\geq G_\rho(\rho_0)=G(\rho_0)-G(\rho)\,,
\end{equation}
and so
\begin{equation}\label{E:26}
e^{2g_\rho(x)}\geq e^{2G(\rho_0)}\cdot e^{-2G(\rho)},\qquad \text{for all}\quad x\in\Omega_{\rho_0}\,.
\end{equation}

Therefore
\begin{equation*}
e^{2G(\rho_0)}\cdot e^{-2G(\rho)}
\int_{\Omega_{\rho_0}}|\psi(x)|^2\,dx \leq
\int_{\Omega_{\rho_0}}\bigl|e^{g_\rho(x)}\psi(x)\bigr|^2\,dx \leq
\int_{\Omega_{2\rho}}\bigl|e^{g_\rho(x)}\psi(x)\bigr|^2\,dx\,,
\end{equation*}
where we used the fact that $2\rho\leq\rho_0$ and so
$\Om_{\rho_0}\subset\Om_{2\rho}$. Now note that $\nabla
g_\rho=\nabla g$, and so   $g_\rho$ satisfies \eqref{E:13} with
the same $E$ and $c$ as $g$. In particular, one can apply
Proposition~\ref{P:1} and obtain
\begin{equation*}
\int_{\Omega_{2\rho}}\bigl|e^{g_\rho(x)}\psi(x)\bigr|^2\,dx\leq
\frac{K(c)}{\rho}\int_{\Omega_\rho\setminus\Omega_{2\rho}}
\biggl(\frac1\rho+\bigl|\nabla
g_\rho(x)\bigr|\biggr)\bigl|e^{g_\rho(x)}\psi(x)\bigr|^2\,dx\,.
\end{equation*}
Since $0<\rho<2\rho<\rho_0\leq d_0/2$, it follows that $\Om_\rho\setminus\Om_{2\rho}\subset\Om\setminus
\Om_{d_0/2}$ and so \eqref{E:22} implies that
\begin{equation*}
\frac{K(c)}{\rho}\int_{\Omega_\rho\setminus\Omega_{2\rho}}
\biggl(\frac1\rho+\bigl|\nabla
g_\rho(x)\bigr|\biggr)\bigl|e^{g_\rho(x)}\psi(x)\bigr|^2\,dx \leq
\frac{\tilde
K(c)}{\rho^2}\int_{\Omega_\rho\setminus\Omega_{2\rho}}\bigl|\psi(x)\bigr|^2\,dx\,,
\end{equation*}
where we also used the fact that, for $x\in\Om_{\rho}\setminus\Om_{2\rho}$,
$$
g_{\rho}(x)\leq G_{\rho}(2\rho)=G(2\rho)-G(\rho)=\int_{\rho}^{2\rho}\! G'(t)\,dt\leq\int_{\rho}^{2\rho}\!
\frac1t\,dt=\log 2.
$$

Putting it all together, we get that
\begin{equation}\label{E:27}
K_2(c,\rho_0)\cdot \rho^2 e^{-2G(\rho)}
\int_{\Omega_{\rho_0}}|\psi(x)|^2\,dx \leq
\int_{\Omega_\rho\setminus\Omega_{2\rho}}\bigl|\psi(x)\bigr|^2\,dx\,.
\end{equation}

Now, let $n\geq1$ be an integer, and set
$$
\rho_n=\frac{1}{2^n}\rho_0\,.
$$
So $2\rho_n=\rho_{n-1}$, and we get
$$
\bigcup_{n=1}^M \bigl(\Om_{\rho_n}\setminus\Om_{2\rho_n}\bigr)
=\bigcup_{n=1}^M \bigl(\Om_{\rho_n}\setminus\Om_{\rho_{n-1}}\bigr)=\Om_{\rho_M}\setminus\Om_{\rho_0}\subset
\Om
$$
So using \eqref{E:27} successively with $\rho=\rho_n$, $1\leq n\leq M$, and summing leads to
\begin{equation}\label{E:28}
\rho_0^2K_2(c,\rho_0)\biggl(\sum_{n=1}^M 4^{-n}
e^{-2G(2^{-n}\rho_0)}\biggr) \int_{\Omega_{\rho_0}}|\psi(x)|^2\,dx
\leq \int_{\Om} \bigl|\psi(x)\bigr|^2\,dx<\infty\,.
\end{equation}
But from condition ($\Sigma.2$) we know that the series
$\sum_{n\geq1}4^{-n} e^{-2G(2^{-n}\rho_0)}$ diverges, and so we
find that
\begin{equation}\label{E:29}
\int_{\Omega_{\rho_0}} \bigl|\psi(x)\bigr|^2\,dx=0.
\end{equation}
But $\rho_0>0$ was arbitrary, and so by taking $\rho_0\to0$ it follows that
\begin{equation}
\int_\Omega \bigl|\psi(x)\bigr|^2\,dx=0\,,
\end{equation}
as claimed.
\end{proof}

\section{Proofs of the main theorems}\label{S:4}

Our strategy in approaching Theorem~\ref{T:1} consists of
combining Agmon-type decay estimates for (weak) eigenfunctions
(see Theorem~\ref{T:2}) with multidimensional Hardy inequalities.
More precisely, for $H$ as above, the fundamental criterion for
self-adjointness tells us that Theorem~\ref{T:1} follows from the
following
\begin{lemma}\label{S4L1}
With the hypotheses of Theorem~\ref{T:1}, there exists an $E<0$
such that for every $\psi\in L^2(\Omega)$ the condition
\begin{equation}\label{E:41}
\langle\psi,(H-E)\varphi\rangle=0,\qquad\text{for every}\quad\varphi\in C_0^\infty(\Omega)
\end{equation}
implies that $\psi\equiv0$.
\end{lemma}

\begin{proof}
In view of Theorem \ref{T:2} the only thing to be proved is that
for $\psi \in L^2(\Omega)$, \eqref{E:4} implies that $\psi $ is a
weak solution of $(H-E)\psi=0 $ i.e. $\psi \in H^1_{loc}(\Omega)$.
This is an interior  regularity result for elliptic equations and
follows from general theory . In our simple setting one can see by
elementary means that $\psi \in H^2_{loc}(\Omega)$. Indeed let
$\tilde \Omega \subset  \Omega $, $dist (\tilde \Omega,
\partial  \Omega) >0$. Then $V \in L^{\infty}(\tilde \Omega)$ and
from $\langle \psi, (H-E)\varphi \rangle =0$ it follows
$$
|\langle \psi, (-\Delta +1)\varphi \rangle |= |\langle \psi, (V-E)\varphi \rangle |
\leq K_{\tilde \Omega, E}\Vert \varphi \Vert
$$
for all $\varphi \in C_0^\infty (\tilde \Omega)$ which via Riesz lemma implies
$$
\langle \psi, (-\Delta +1)\varphi \rangle =\langle \Phi,\varphi \rangle
$$
for some $\Phi \in L^2 (\tilde \Omega)$. This means that the
distribution $(-\Delta +1)\Psi$ on $C_0^\infty (\tilde \Omega)$ is
represented b y a $L^2 (\tilde \Omega)$ function and the proof is
finished.
\end{proof}

The following multidimensional Hardy inequality will allow us to
complete the proof of our main theorem:
\begin{theorem}[Multidimensional Hardy Inequality]\label{T:Hardy}
Let $\Om\subset\IR^n$ be a bounded open set with
$C^2$-smooth boundary. Then there exists a constant $A=A(\Om)\in\IR$ such that
\begin{equation}\label{E:S4eq1}
\frac14\int_\Om\!\!\frac{|\varphi(x)|^2}{d(x)^2}\,dx\leq
\int_\Om\!\! |\nabla\varphi(x)|^2\,dx+A\|\varphi\|^2
\end{equation}
for every $\varphi\in C_0^\infty(\Omega)$.
\end{theorem}
This particular form of the Hardy inequality in domains in $\IR^n$ can be found, for example,
in \cite{Dav}.

Now the proof of Theorem~\ref{T:1} follows very quickly.
\begin{proof}[Proof of Theorem~\ref{T:1}]
From the fundamental criterion for self-adjointness (via
Lemma~\ref{S4L1}) and the Agmon-type Theorem~\ref{T:2}, we
conclude that what we must show in order to complete the proof is
that there exist $E\in\IR$, as well as $c>0$ and a function
$g(x)=G(d(x))$ with $G$ satisfying ($\Sigma$) such that
\begin{equation}\label{E:S4eq2}
\langle \varphi,(H-E)\varphi \rangle -\int_\Omega |\nabla g(x)|^2|\varphi(x)|^2\,dx\geq c\|
\varphi\|^2\,,
\end{equation}
for all $\varphi\in C_0^\infty(\Om)$.

Recall that under the hypotheses of Theorem~\ref{T:1}, the
potential $V=V_1+V_2$ with $V_2\in L^\infty(\Om)$ and
$$
V_1(x)\geq G'\bigl(d(x)\bigr)^2-\frac14\cdot\frac{1}{d(x)^2}\,,
$$
for some $G$ satisfying ($\Sigma$). Using exactly this $G$ to
define the $g$ we need, and applying the result of the
multidimensional Hardy inequality above, we get that for $E\in\IR$
\begin{equation*}
\begin{aligned}
\langle \varphi,(H-E)\varphi \rangle &-\int_\Omega |\nabla g(x)|^2|\varphi(x)|^2\,dx\\
&\geq\int_\Om\Bigl(V_1(x)-G'(d(x))^2+\frac{1}{4d(x)^2}\Bigr)\cdot|\varphi(x)|^2\,dx\\
&\quad+\bigl(-\|V_2\|_{L^\infty}-A-E\bigr)\|\varphi\|^2\\
&\geq \bigl(-\|V_2\|_{L^\infty}-A-E\bigr)\|\varphi\|^2
\end{aligned}
\end{equation*}
On the way we have used the fact that $|\nabla g(x)|^2 \leq G'\bigl(d(x)\bigr)^2$.
So choosing, for example,
\begin{equation}
E=-\|V_2\|_{L^\infty}-A-1
\end{equation}
leads to \eqref{E:S4eq2} being satisfied with $c=1$. This is exactly what we needed, and concludes our proof.
\end{proof}

\begin{proof}[Proof of Theorem~\ref{T:p}]
As already explained in Section ~\ref{S:2}, Theorem ~\ref{T:p}
follows directly from Theorem~\ref{T:1} and a choice of function $G$
which for small $t$ coincides with (see \eqref{Ex:31}):
\begin{equation*}
\ln t+\frac12\cdot\sum_{j=2}^p L_j(t)+\int_t \tilde f(u)\,du\,,
\end{equation*}
where we recall that the functions $L_j$ were defined in
\eqref{E:DefnLp}, and $\tilde f$, which is to be found, must
satisfy \eqref{Ex:24}.

More precisely, let
\begin{equation}\label{scL}
{\mathcal L}_p(t) = \sum_{k=2}^p\Bigl(\prod_{j=1}^{k-1}L_j(t)\Bigr)^{-1}
\end{equation}
defined for $0<t<e_p^{-1}$. Notice that $\lim_{t \rightarrow
0}\mathcal{L}_p(t) =0$ and moreover $\frac{\mathcal{L}_p(t)^2}{t}$
is integrable at zero. So if we define
\begin{equation}\label{tildef}
\tilde f(t)=
\left\{
  \begin{array}{ll}
    f(t)+\frac{1}{4t}\mathcal L_p(t)^2, & \text{for}\,\, 0\leq t\leq e_p^{-1} \\
    0, & \text{for}\,\, t\geq e_p^{-1},
  \end{array}
\right.
\end{equation}
then it satisfies \eqref{Ex:24}. Let now $h(t)$ be a smooth function with the properties:
\begin{equation}\label{h}
h(t)=
\left\{
  \begin{array}{ll}
    t, & \text{for}\,\, 0\leq t\leq \frac{d_0}{2} \\
    \frac34 d_0, & \text{for}\,\, t\geq d_0,
  \end{array}
\right.
\end{equation}
and $0 <h'(t)\leq 1$ for all $0<t<d_0$. Here $d_0 \leq \min \{
e_p^{-1}, d_{\Omega} \}$ and in addition is suficiently small such
that for $t\in (0, d_0)$
\begin{equation}\label{d0<}
1-\frac{1}{2}\mathcal L_p(t)-t \tilde f(t) \geq \frac{2}{3}.
\end{equation}
We claim that
\begin{equation}\label{Gp}
G_p(t)=\ln h(t)+\frac12\cdot\sum_{j=2}^p L_j(h(t))+\int_{h(t)} \tilde f(u)\,du\,,
\end{equation}
satisfies all the needed conditions.

To check that $G_p$ satisfies ($\Sigma$), first note that, for any $k\geq1$,
\begin{equation*}
L_k'(t)=-\frac{1}{t}\Bigl(\prod_{j=1}^{k-1}L_j(t)\Bigr)^{-1}\,,
\end{equation*}
and so, for $t\in (0,d_0)$
\begin{equation}\label{Gpprim}
G_p'(t)=\frac{1}{h(t)}\cdot\Bigl[1-\frac12\mathcal L_p(h(t))-h(t)\tilde f(h(t))\Bigr]h'(t)\,,
\end{equation}
while for $t \geq d_0$, $G_p'(t)=0$. Then ($\Sigma.1$)
follows from \eqref{Gpprim} , \eqref{d0<} and the properties of $h(t)$.

To check ($\Sigma.2$), note that from \eqref{Gp} for $t
<\frac{d_0}{2}$ (take into account that $d_0 <1$ and for $t
<\frac{d_0}{2}$, $h(t)=t$)
\begin{equation*}
e^{-2G_p(2^{-n}\rho_0)}\geq 4^n \rho_0^{-2}e^{-2\int_{0+} \tilde f(u)\,du}\cdot \Bigl(\prod_{j=1}^{p-1}
L_j (2^{-n}\rho_0)\Bigr)^{-1}
\end{equation*}
If we define, for $x\in\IR$ large enough, the log-log functions $\ln_0(x)=x$,
$\ln_k(x)=\ln\bigl(\ln_{k-1}(x)\bigr)$, then note that for all $1 \leq j\leq p-1$ and
$n\geq N(\rho_0)= \frac{1}{1-\ln 2}\ln \rho_{0}^{-1}$ (remember that $2\rho_0 <
e_{p}^{-1}$)
\begin{equation*}
L_j(2^{-n}\rho_0)=\ln_j(2^n\rho_0^{-1})=\ln_{j-1}(n\ln2+\ln \rho_0^{-1})\leq \ln_{j-1} n\,.
\end{equation*}
But then
\begin{equation*}
\sum_{n=0}^\infty 4^{-n} e^{-2G_p(2^{-n}\rho_0)} \geq \text{const.} \sum_{n=N(\rho_0)}^\infty
\frac{1}{n\ln(n)\ln_2(n)\cdots\ln_{p-2}(n)}=+\infty\,,
\end{equation*}
where the divergence of the latter series is an elementary
consequence of the integral test. Since $ \sup_{t\geq
\frac{d_0}{2}} G'_p(t)^2 < \infty$, in view of the remark
following Theorem~\ref{T:p}, all that remains to be done in order
to apply Theorem~\ref{T:1} is to show that for $t\in
(0,\frac{d_0}{2})$ \eqref{E:3p} implies \eqref{E:3} with
$G(t)=G_p(t)$. Taking into account \eqref{Gpprim} it is sufficient
to check that for $t\in (0,\frac{d_0}{2})$:

\begin{equation}
\frac{1}{t^2}-\frac{1}{t^2}\mathcal L_p(t)-\frac{1}{t}f(t) \geq \frac{1}{t^2} (1-\frac{1}{2}
\mathcal L_p(t)-t\tilde f(t))^2
\end{equation}
Doing the algebra one gets the condition
$$
-\frac{f(t)}{t} \geq -2\frac{\tilde f(t)}{t} +\frac{\mathcal L_p(t)^2}{4t^2} +
(t\tilde f(t)+\mathcal L_p(t))\frac{\tilde f(t)}{t}
$$
Now taking into account  \eqref{tildef} and that from \eqref{d0<},
 for $t\in (0,\frac{d_0}{2})$, $\mathcal L_p \leq \frac 23$,
$t\tilde f(t) \leq \frac13$ one has
$$
-2\frac{\tilde f(t)}{t} +\frac{\mathcal L_p(t)^2}{4t^2} +(t\tilde f(t)+\mathcal L_p(t))\frac{\tilde f(t)}{t}
\leq -\frac{\tilde f(t)}{t} +\frac{\mathcal L_p(t)^2}{4t^2} = -\frac{ f(t)}{t},
$$
and the proof is finished.
\end{proof}

Finally we turn to the proof of our optimality theorem:
\begin{proof}[Proof of Theorem~\ref{T:pOptimal}]
In order to achieve this, we will work in 1 dimension, on the interval $(0,1)$,
 and construct such a potential close to 0. In this case, let $\alpha\in\IR$ and consider the wave function
\begin{equation}\label{E:31}
\psi_{p,\alpha}(x)=x^{-\frac12}\cdot\biggl(\prod_{j=1}^{p-1} L_j(x)\biggr)^{-\frac12}\cdot L_p(x)^\alpha\,.
\end{equation}
First note that $\psi_{p,\alpha}$ grows as $x\to0+$ for all $\alpha\in\IR$, but that
$$
\int_{0+} \psi_{p,\alpha}^2(x)\,dx=\infty\qquad\Longleftrightarrow\qquad \alpha\geq-\frac12\,.
$$
A direct calculation shows that
\begin{equation}\label{E:32}
\psi_{p,\alpha}''(x)=V_{p,\alpha}(x)\psi_{p,\alpha}(x)\,,
\end{equation}
with
\begin{equation}\label{E:33}
V_{p,\alpha}(x)=\frac34\cdot\frac{1}{x^2} - \frac{1}{x^2}\cdot\sum_{j=1}^{p-1}\biggl(\prod_{k=1}^j L_k(x)
\biggr)^{-1}+\bigl(2\alpha+o(1)\bigr)\cdot\frac{1}{x^2}\cdot\biggl(\prod_{j=1}^p L_j(x)\biggr)^{-1}
\end{equation}
where the $o(1)$ comes from a sum of terms which are of lower
order (in the same spirit as in the previous proof). In this case,
they are
\begin{align*}
&\frac14\cdot\frac{1}{x^2} \cdot\biggl(\sum_{j=1}^{p-1}\biggl(\prod_{k=1}^j L_k(x)\biggr)^{-1}\biggr)^2+\frac{\alpha^2}{x^2}\cdot \biggl(\prod_{j=1}^p L_j(x)\biggr)^{-2}\\
                         &-\frac{\alpha}{x^2}\cdot\biggl(\sum_{j=1}^{p-1}\biggl(\prod_{k=1}^j L_k(x)\biggr)^{-1}\biggr)\cdot\biggl(\prod_{j=1}^p L_j(x)\biggr)^{-1}\\
                         &+\frac12\cdot\frac{1}{x^2}\sum_{j=1}^{p-1} \sum_{k=1}^j \biggl(\prod_{l=1}^j L_l(x)\biggr)^{-1}\biggl(\prod_{l=1}^k L_l(x)\biggr)^{-1}\\
                         &-\frac{\alpha}{x^2}\sum_{k=1}^p \biggl(\prod_{l=1}^p L_l(x)\biggr)^{-1}\biggl(\prod_{m=1}^k L_m(x)\biggr)^{-1}
\end{align*}
Further note that the other (decreasing at $0+$) solution of
\begin{equation*}
\phi_{p,\alpha}''(x)=V_{p,\alpha}(x)\phi_{p,\alpha}(x)
\end{equation*}
is given by the usual relation
\begin{equation*}
\phi_{p,\alpha}(x)=\psi_{p,\alpha}(x)\cdot\int_0^x \frac{1}{\psi_{p,\alpha}^2(y)}\,dy\,.
\end{equation*}
Since $\psi_{p,\alpha}^{-2}(y)\sim y^{1-\epsilon}$ as $y\to0+$ for
any given $\epsilon>0$, we see that $\phi_{p,\alpha}(x)\to0$ as
$x\to0+$, and so in particular $\phi_{p,\alpha}$ and
$\psi_{p,\alpha}$ are indeed two independent solutions. But for
$\alpha<-\frac12$, they are both in $L^2(0+)$ and so we are in the
limit-circle case and
\begin{equation*}
H_{p,\alpha}=-\Delta+V_{p,\alpha}
\end{equation*}
is \textit{not} essentially self-adjoint on $(0,1)$. But this is
exactly the type of potential we were looking for: given a
constant $c>1$, pick an $\alpha<-\frac{c}{2}<-\frac12$. Thus
$H_{p,\alpha}$ is not essentially self-adjoint, but for $x$ close
enough to the boundary $\partial\Omega$, equation~\eqref{E:33}
together with our choice of $\alpha$ implies that $V_{p,\alpha}$
satisfies \eqref{E:4p}, as claimed in the theorem.

Finally, the potentials $V_{p,\alpha}$ can also be used in several space dimensions to construct counterexamples
\end{proof}


\end{document}